%% file: root.tex
\documentclass[letterpaper, 10 pt, conference]{ieeeconf}  

\IEEEoverridecommandlockouts                              
\usepackage{textcomp}
\usepackage{stfloats}
\usepackage{url}
\usepackage{verbatim}
\usepackage{graphicx}
\usepackage[sort,compress,nospace]{cite}
\usepackage{threeparttable}
\usepackage{tablefootnote}

\usepackage{amsthm}
\usepackage[dvipsnames]{xcolor}
\usepackage{amsmath,amsfonts}
\usepackage{algorithmic}
\usepackage{algorithm}
\usepackage{array}
\usepackage{siunitx}
\usepackage{epstopdf}
\usepackage[caption=false, font=footnotesize]{subfig}
\usepackage{svg}
\usepackage{multirow}
\usepackage{booktabs}
\usepackage{comment}
\usepackage{cuted}
\usepackage{lipsum}
\usepackage{txfonts}

\newtheorem{theorem}{Theorem}

\newtheorem{proposition}{Proposition}

\newtheorem{remark}{Remark}

\title{\LARGE \bf
Linear-Quadratic Gaussian Games with Distributed Sparse Estimation
}

\author{Tianyu Qiu$^1$, Filippos Fotiadis$^1$, Xinjie Liu$^1$, Christian Ellis$^{1,2}$, Jesse Milzman$^2$, Wesley Suttle$^2$,\\ Ufuk Topcu$^1$ and David Fridovich-Keil$^1$
\thanks{This work was sponsored by the ARL under Cooperative Agreements W911NF-23-2-0011 and W911NF-25-2-0021, by NSF under grants 2211548, 2336840, and 2409535, by ONR under grant N00014-24-1-2797, and by AFOSR under grant  FA9550-22-1-0403.}
\thanks{T. Qiu, F. Fotiadis, X. Liu, C. Ellis, U. Topcu, and D. Fridovich-Keil are with the Oden Institute for Computational Engineering and Sciences, The University of Texas at Austin, TX, 78712, USA. 
(Email:  {\tt \{tianyuqiu, ffotiadis, xinjie-liu, utopcu, dfk\}@utexas.edu, christian.ellis@austin.utexas.edu}).
}%
\thanks{
C. Ellis, J. Milzman, and W. Suttle are with the DEVCOM Army Research Laboratory, MD, 20783-1138, USA. 
(Email:  {\tt jesse.m.milzman.civ@army.mil, suttlewes@gmail.com}).
}
}

\begin{document}

\maketitle
\thispagestyle{empty}
\pagestyle{empty}

\begin{abstract}
\input{abstract}
\end{abstract}

\section{Introduction}
\input{introduction}


\section{Linear-Quadratic Gaussian Games with Individual Observations}
\input{lqgg}

\section{Distributed Sparse Estimation Gain Design}
\input{sensor_selection}


\section{Numerical Simulations}
\input{experiment}

\section{Conclusions and Future Work}
\input{conclusion}








\bibliographystyle{IEEEtran}
\bibliography{IEEEabrv,reference}

\end{document}

%% file: abstract.tex
Linear-quadratic Gaussian games provide a framework for modeling strategic interactions in multi-agent systems, where agents must estimate system states from noisy observations while also making decisions to optimize a quadratic cost. However, these formulations usually require agents to utilize the full set of available observations when forming their state estimates, which can be unrealistic in large-scale or resource-constrained settings. In this paper, we consider linear-quadratic Gaussian games with sparse interagent observations. To enforce sparsity in the estimation stage, we design a distributed estimator that balances estimation effectiveness with interagent measurement sparsity via a group lasso problem, while agents implement feedback Nash strategies based on their state estimates. We provide sufficient conditions under which the sparse estimator is guaranteed to trigger a corrective reset to the optimal estimation gain, ensuring that estimation quality does not degrade beyond a level determined by the regularization parameters. Simulations on a formation game show that the proposed approach yields a significant reduction in communication resources consumed while only minimally affecting the nominal equilibrium trajectories.

%% file: introduction.tex
\label{sec:introduction}

Dynamic games are widely used to model strategic interactions in multi-agent systems, where agents are assumed to share a common representation of state uncertainty.
In many applications, however, agents operate under severe resource constraints. For instance, large-scale robot swarms may rely on limited onboard computation and energy resources, making it impractical to process or exchange all available measurements. Sparsity in sensing design thus becomes essential
for efficient resource allocation. 

Recent works \cite{chahine2023local, liu2024policies, qiu2025psn} 
have explored sparsity in deterministic games.
These works propose sparsity-promoting game formulations to selectively prune less influential interactions while preserving effective overall agent performance.
However, deterministic games simplify common multi-agent scenarios by ignoring process disturbances and noise, which are inevitable in real-world applications. 

Classically, linear-quadratic Gaussian (LQG) games \cite{shapley1953stochastic}  extend dynamic game theory to stochastic settings, where agents interact under process and observation uncertainty. LQG games have been studied extensively and are fundamentally more complex than deterministic LQ games studied in \cite{liu2024policies, SmoothGameTheory} due to the coupling between control and estimation. Specifically, the separation principle~\cite{welch1995introduction,van2011lqg,bertsekas2012dynamic} does not generally hold for these games, and equilibrium strategies depend on the underlying information pattern. As a result, much of the literature on LQG games focuses on characterizing equilibrium solutions under various information structures \cite{pachter2013static,pachter2016linear,gupta2014common,gupta2016dynamic,vasal2016signaling,heydaribeni2019linear,hambly2023linear}. However, these works have not investigated how sparsity in estimation can be incorporated in LQG games.

\begin{figure}[!t]
    \centering
    \includegraphics[width=0.99\linewidth]{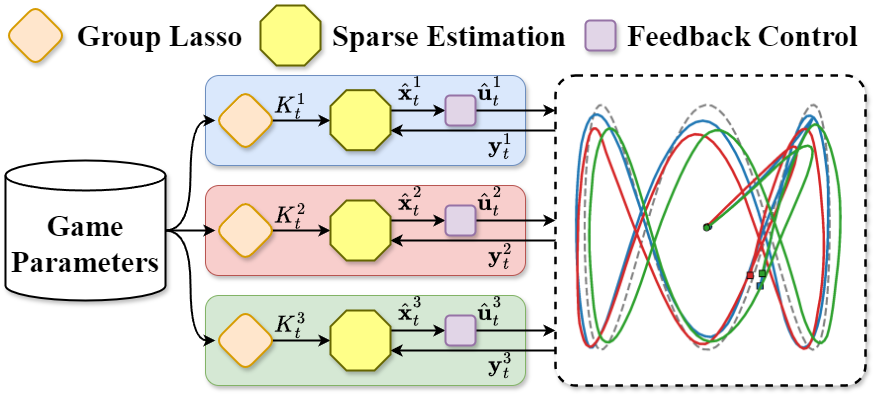}
    \caption{Feedback LQG game play with distributed sparse estimation in a three-robot formation game.}
    \label{fig:framework}
    \vspace{-4ex}
\end{figure}

A natural way to introduce sparsity in LQG games is through their estimation stage. In particular, sparsity in estimation can be interpreted as a sensor selection problem, where agents selectively utilize certain observation components when forming their state estimates. Existing work on sensor selection and sparsity-promoting estimation has primarily focused on single-agent or centralized systems \cite{tzoumas2020lqg,fotiadis2024input,fotiadis2026input}. Related work on distributed estimation also exists \cite{zhong2024sparse}. However, the distributed estimation problem for multi-agent dynamic games brings two difficulties: \romannumeral1) agents interact strategically, and \romannumeral2) agents are subject to distinct measurements, making classical sensor selection approaches inapplicable. Consequently, distributed sparse estimation design in stochastic games remains largely unexplored. 

In this work, we study how sparsity can be incorporated into LQG games with individual observations. Towards this end, our contributions are as follows: 
\romannumeral1) We first design a distributed state estimator for LQG games that operates using each agent's local measurements while agents follow feedback Nash equilibrium strategies \cite{SmoothGameTheory, liu2024policies}. 
\romannumeral2) Building on this estimator, we introduce a sparsity-promoting distributed estimation framework for LQG games based on group lasso, which promotes structured sparsity across interagent observation groups. 
\romannumeral3) We provide sufficient conditions showing that if the estimation error covariance grows beyond a level determined by the regularization parameters, the sparse estimator is guaranteed to reset to the optimal estimation gain, thereby preventing unbounded degradation of estimation quality.
\romannumeral4) Finally, we validate the proposed framework through simulations on a three-robot formation game, as illustrated in Fig. \ref{fig:framework}, demonstrating the trade-off between estimation sparsity and trajectory performance in noisy multi-agent environments.

%% file: lqgg.tex
\label{sec:lqgg_ie}
\subsection{Linear-Quadratic Gaussian Games}
We consider an $N$-player discrete-time finite-horizon linear-quadratic Gaussian (LQG) game.  In this setting, each agent ${i\in[N]\coloneqq\{1,\dots,N\}}$\footnote{In this work, we denote $[Z]\coloneqq\{1,\dots,Z\}$ for any positive integer $Z$.} is subject to the linear dynamics
\begin{equation*}
        x_{t+1}^i\coloneqq A_t^ix_t^i+B_t^iu_t^i+w_t^i,
\end{equation*}
with state ${x_t^i\in\mathbb{R}^{n_i}}$, control input ${u_t^i\in\mathbb{R}^{m_i}}$, and system matrices ${A_t^i\in\mathbb{R}^{n_i\times n_i}}$ and ${B_t^i\in\mathbb{R}^{n_i\times m_i}}$. ${w_t^i\in\mathbb{R}^{n_i},w_t^i\sim\mathcal{N}(0,W_t^i)}$ is Gaussian process noise with covariance $W_t^i\in\mathbb{R}^{n_i\times n_i}\succeq0$. We further denote the joint state of all agents as ${\textbf{x}_t\coloneqq[x_t^{1\top},\dots,x_t^{N\top}]^\top\in\mathbb{R}^n}$ and the joint control input of all agents as ${\textbf{u}_t\coloneqq[u_t^{1\top},\dots,u_t^{N\top}]^\top\in\mathbb{R}^m}$ with joint state dimension ${n\coloneqq\sum_{i=1}^Nn_i}$ and joint control dimension ${m\coloneqq\sum_{i=1}^Nm_i}$.

The objective of each agent $i$ is to minimize a quadratic performance criterion over a finite time horizon $T$, subject to the joint dynamics of all agents. Agent $i$'s problem can be formulated as
\begin{subequations}
    \begin{align}
        \min_{\textbf{u}^i}\ &\sum_{t=1}^T\frac{1}{2}\mathbb{E}\left[(\textbf{x}_t^\top Q_t^i+2q_t^{i\top})\textbf{x}_t+(\textbf{u}_t^{\top}R_t^i+2r_t^{i\top})\textbf{u}_t\right]\notag\\
        &+\frac{1}{2}\mathbb{E}\left[(\textbf{x}_{T+1}^\top Q_{T+1}^i+2q_{T+1}^{i\top})\textbf{x}_{T+1}\right],\label{eqn:lqgg_cost}\\
        \text{s.t.}\  \textbf{x}_{t+1}&=A_t\textbf{x}_t+B_t\textbf{u}_t+\textbf{w}_t,\label{eqn:lqgg_dyn}\\
        \textbf{x}_1&\sim\mathcal{N}(\Bar{\textbf{x}}_1,\Sigma_{\textbf{x},1}^-),\label{eqn:lqgg_initial_cond}
    \end{align}
    \label{eqn:DTLQGG}%
\end{subequations}
with $A_t\coloneqq\text{blockdiag}(\{A_t^i\}_{i=1}^N),B_t\coloneqq\text{blockdiag}(\{B_t^i\}_{i=1}^N)$ and joint process noise $\textbf{w}_t\coloneqq[w_t^{1\top},\dots,w_t^{N\top}]^\top\in\mathbb{R}^n$, following
\begin{equation*}
    \begin{aligned}
        &\textbf{w}_t\sim\mathcal{N}(\textbf{0},W_t),\ W_t\coloneqq\text{blockdiag}(\{W_t^i\}_{i=1}^N)\in\mathbb{R}^{n\times n}.
    \end{aligned}
\end{equation*}
The joint initial state $\textbf{x}_1$ follows a prior Gaussian distribution with mean ${\Bar{\textbf{x}}_1}$ and covariance ${\Sigma_{\textbf{x},1}^-\succeq0}$.
The $N$ agents' coupled versions of problem \eqref{eqn:DTLQGG} form an LQG game.

\subsection{Individual Observation Model}
In LQG games, each agent $i$ usually does not have access to all agents' full states $\textbf{x}_t$. Instead, it has access to some partial, noisy individual observations ${\textbf{y}_t^i\in\mathbb{R}^{p_i}}$, given by
\begin{equation}
        \textbf{y}_t^i\coloneqq C_t^i\textbf{x}_t+v_t^i,\  v_t^i\in\mathbb{R}^{p_i},v_t^i\sim\mathcal{N}(0,V_t^i),\ V_t^i\in\mathbb{R}^{p_i\times p_i}\succeq0,\label{eqn:ind_obs}
\end{equation}
where ${C_t^i\in\mathbb{R}^{p_i\times n}}$ denotes the individual observation matrix of agent $i$, and {$v_t^i$} denotes the individual Gaussian observation noise. 
We also denote $\textbf{v}_t\coloneqq[v_t^{1\top},\dots,v_t^{N\top}]^\top\in\mathbb{R}^p$, following
\begin{equation*}
    \begin{aligned}
        &\textbf{v}_t\sim\mathcal{N}(\textbf{0},V_t),\ V_t\coloneqq\text{blockdiag}(\{V_t^i\}_{i=1}^N)\in\mathbb{R}^{p\times p},\  p=\sum_{i=1}^Np_i.
    \end{aligned}
\end{equation*}
In this paper, we assume that $\textbf{y}_t^i$ can be obtained from $P_i$ individual sensors and thus partitioned into $\textbf{y}_t^{i1},\dots,\textbf{y}_t^{iP_i}$. It is also common that $C_t^i=I$, implying that $\textbf{y}_t^i$ may comprise interagent observations.

\subsection{Control, Estimation, and Information Patterns}
There is generally no closed-form solution for the feedback Nash equilibrium in LQG games with individual observations \cite{witsenhausen1968counterexample}. A common low-complexity approximation assumes that each agent $i$ employs the Nash control gains $\Gamma_t\coloneqq[\Gamma_t^1,\cdots,\Gamma_t^N]^\top$ and $\alpha_t\coloneqq[\alpha_t^{1\top},\cdots,\alpha_t^{N\top}]^\top$ obtained from the LQG game with process noise but without measurement noise \cite{liu2024policies, SmoothGameTheory}. Each agent then estimates the joint state $\hat{\textbf{x}}_t^i$ from its individual observations $\textbf{y}_t^i$ with estimation gain $K_t^i$, and substitutes this estimate into its control input $u_t^i$, resulting in:
\begin{equation}
    u_t^i=-\Gamma_t^i\hat{\textbf{x}}_t^i-\alpha_t^i,\label{eqn:actual_single_ctrl}
\end{equation}
whose performance relies on the accuracy of $\hat{\textbf{x}}_t^i$. To characterize this accuracy, we define agent $i$'s individual estimation error $\textbf{e}_t^i$ and the joint estimation error across agents $\textbf{e}_t$ as
\begin{equation*}
    \textbf{e}_t^i\coloneqq\textbf{x}_t-\hat{\textbf{x}}_t^i,\ \textbf{e}_t\coloneqq[\textbf{e}_t^{1\top},\cdots,\textbf{e}_t^{N\top}]^\top,
\end{equation*}
and their covariances as $\Sigma_t^i\coloneqq \text{Cov}(\textbf{e}_t^i)$, $\Sigma_t\coloneqq \text{Cov}(\textbf{e}_t)$. 

The central question we investigate in this paper is how to design these gains $K_t^i$ to enforce interagent measurement sparsity, while balancing that against control performance. To that end, we assume the following information pattern:
\begin{enumerate}
    \item The system matrices $A_t,B_t$, observation model $C_t^i$, process and measurement noise covariance $W_t, V_t$, control gains $\Gamma_t,\alpha_t$, and distributed estimation gain $K_t^i$ are common knowledge among agents. 
    \item The individual observations $\textbf{y}_t^i$, together with individual estimation $\hat{\textbf{x}}_t^i$ computed via $\textbf{y}_t^i$, and feedback control $\textbf{u}_t^i$ computed via $\hat{\textbf{x}}_t^i$ are \textbf{not} shared among agents.
\end{enumerate}
This information pattern captures how agents' decisions depend upon their limited access to information, particularly when agents' sensing of one another is \textit{\textbf{sparse}}.

%% file: sensor_selection.tex
\label{sec:obs_model_selection}

In this section, we design sparse estimation gains $K_t^i$ for each agent, subject to the aforementioned information pattern. We first derive the propagation of estimation error and covariance under the distributed information structure, then obtain the unconstrained optimal estimation gain as a reference solution, and finally formulate a group lasso optimization that sparsifies this gain by penalizing reliance on uninformative sensors.

\subsection{Propagation of Joint Estimation Error and Covariance}
\label{sec:err_cov_propgation}
\textbf{Joint Control Input:} A consequence of the assumed information pattern is that each agent $i$ cannot reconstruct the actual joint control $\textbf{u}_t$ from \eqref{eqn:actual_single_ctrl} applied to the joint system, because agents do not share their state estimates. Instead, each agent can only reconstruct an \textit{estimated} $\hat{\textbf{u}}_t^i$ based on its estimation of the joint state $\hat{\textbf{x}}_t^i$. This estimate is given by
\begin{equation}
    \hat{\textbf{u}}_t^i=\begin{bmatrix}
            \hat{u}_t^{i1^\top}&\cdots&\hat{u}_t^{iN^\top}
        \end{bmatrix}^\top=-\Gamma_t\hat{\textbf{x}}_t^i-\alpha_t=-\Gamma_t(\textbf{x}_t-\textbf{e}_t^i)-\alpha_t,
    \label{eqn:estimated_control}%
\end{equation}
where agent $i$'s actual control $u_t^i$  is the $i$-th block of $\hat{\textbf{u}}_t^i$, i.e., $u_t^i=\hat{u}_t^{ii}$. Then, the actual joint control $\textbf{u}_t$ applied to the system is the concatenation of all agents' actual control $u_t^i$, given as
\begin{equation}
    \begin{aligned}
        \textbf{u}_t&=\begin{bmatrix}
            \hat{u}_t^{11^\top}&\cdots&\hat{u}_t^{NN^\top}
        \end{bmatrix}^\top=\sum_{j=1}^NE^j\textbf{u}_t^j\\
        &=-\alpha_t-\sum_{j=1}^NE^j\Gamma_t(\textbf{x}_t-\textbf{e}_t^j)=-\alpha_t-\Gamma_t\textbf{x}_t+\sum_{j=1}^NE^j\Gamma_t\textbf{e}_t^j,
    \end{aligned}
    \label{eqn:actual_control}
\end{equation}
with $E^i\in\mathbb{R}^{m\times m}$ being a block diagonal matrix with its $(i,i)$-th block being $I_{m_i}$, and the other blocks being zero matrices. 

\textbf{Joint State Propagation:} Given the joint control input \eqref{eqn:actual_control},  the joint state $\textbf{x}_{t+1}$ evolves according to
\begin{equation}
    \begin{aligned}
        \textbf{x}_{t+1}&=A_t\textbf{x}_t+B_t\textbf{u}_t+\textbf{w}_t\\
        &=(A_t-B_t\Gamma_t)\textbf{x}_t+\sum_{j=1}^NB_tE^j\Gamma_t\textbf{e}_t^j-B_t\alpha_t+\textbf{w}_t.
    \end{aligned}
    \label{eqn:x_update}
\end{equation}
\textbf{Linear Estimation Propagation:} 
Each agent $i$ then updates the estimate $\hat{\textbf{x}}_{t+1}^i$ according to its individual observations $\textbf{y}_{t+1}^i$ and the estimate of control $\hat{\textbf{u}}_t^i$ via
\begin{subequations}
    \begin{align}
        \hat{\textbf{x}}_{t+1}^{i-}&=A_t\hat{\textbf{x}}_t^i+B_t\hat{\textbf{u}}_t^i,\label{eqn:prior_x_update}\\
        \hat{\textbf{x}}_{t+1}^i&=\hat{\textbf{x}}_{t+1}^{i-}+K_{t+1}^i\left(\textbf{y}_{t+1}^i-C_{t+1}^i\hat{\textbf{x}}_{t+1}^{i-}\right),\label{eqn:post_x_update}
    \end{align}
    \label{eqn:est_x_update}%
\end{subequations}
where $K_t^i$ is to be designed in the following subsections. Moreover, subtracting \eqref{eqn:prior_x_update} and \eqref{eqn:post_x_update} from \eqref{eqn:x_update}, we obtain each agent $i$'s individual estimation error $\textbf{e}_{t+1}^i$ as
\begin{subequations}
    \begin{align}
        \textbf{e}_{t+1}^{i-}&=\textbf{x}_{t+1}-\hat{\textbf{x}}_{t+1}^{i-}=A_t(\textbf{x}_t-\hat{\textbf{x}}_t^i)+B_t(\textbf{u}_t-\hat{\textbf{u}}_t^i)+\textbf{w}_t\notag\\
        &=A_t\textbf{e}_t^i+B_t\left(\sum_{j=1}^NE^j\Gamma_t\textbf{e}_t^j-\Gamma_t\textbf{e}_t^i\right)+\textbf{w}_t,\label{eqn:ind_err_prior_update}\\
        \textbf{e}_{t+1}^i&=\textbf{x}_{t+1}-\hat{\textbf{x}}_{t+1}^i=\textbf{e}_{t+1}^{i-}-K_{t+1}^i\left(C_{t+1}^i\textbf{e}_{t+1}^{i-}+v_{t+1}^i\right).\label{eqn:ind_err_post_update}
    \end{align}
    \label{eqn:est_update}%
\end{subequations}
\vspace{-4ex}
\begin{remark}
A key observation from \eqref{eqn:est_update} is that, unlike classical state estimation where error depends only on system dynamics and noise, the error in LQG games with individual observations is also driven by the mismatch between agent $i$'s estimate of the joint control $\hat{\textbf{u}}_t^i$ and the actual $\textbf{u}_t$, and consequently leads to a joint update for the estimation error.
\end{remark}

\textbf{Covariance Update:} Finally, concatenating \eqref{eqn:ind_err_prior_update} and \eqref{eqn:ind_err_post_update}, the joint estimation error $\textbf{e}_{t+1}$ and covariance $\Sigma_{t+1}$ evolve as
\begin{subequations}
    \begin{align}
        \textbf{e}_{t+1}^-&=(\mathcal{A}_t+\mathcal{B}_t)\textbf{e}_t+\textbf{w}_t,\\
        \textbf{e}_{t+1}&=(I_{n^2}-\mathcal{K}_{t+1})\textbf{e}_{t+1}^--\mathbb{K}_{t+1}\textbf{v}_{t+1},\\
        \Sigma_{t+1}^-&=(\mathcal{A}_t+\mathcal{B}_t)\Sigma_t(\mathcal{A}_t+\mathcal{B}_t)^\top+W_t\label{eqn:prior_cov_update},\\
        \Sigma_{t+1}&=(I_{n^2}-\mathcal{K}_{t+1})\Sigma_{t+1}^-(I_{n^2}-\mathcal{K}_{t+1})^\top+\mathbb{K}_{t+1}V_{t+1}\mathbb{K}_{t+1}^\top,\label{eqn:post_cov_update}
    \end{align}
\end{subequations}
\begin{equation*}
    \begin{aligned}
        &\mathcal{A}_t=\text{blockdiag}(\{A_t-B_t\Gamma_t\}_{i=1}^N),\ 
        \mathcal{B}_t=\begin{bmatrix}
                B_tE^1\Gamma_t & \cdots & B_tE^N\Gamma_t\\
                \vdots & \ddots & \vdots \\
                B_tE^1\Gamma_t & \cdots & B_tE^N\Gamma_t
            \end{bmatrix},\\
        &\mathcal{K}_{t+1}=\text{blockdiag}(\{K_{t+1}^iC_{t+1}^i\}_{i=1}^N),\
        \mathbb{K}_{t+1}=\text{blockdiag}(\{K_{t+1}^i\}_{i=1}^N).
    \end{aligned}
    \label{eqn:augmented_system_matrices}
\end{equation*}
This relation explicitly shows how $\textbf{e}_{t+1}$ and $\Sigma_{t+1}$ are affected by the distributed estimation gains, which we will design sparsely in the following subsections.

\subsection{Distributed Optimal Estimation Design}
\label{sec:distributed_est_gain}

We now derive the unconstrained optimal estimation gain for each agent, which will serve as the reference solution for the sparse design in Section~\ref{sec:sparse_est_group_lasso}. Specifically, we minimize agent $i$'s individual estimation error covariance $\Sigma_t^i$ with respect to its estimation gain $K_t^i$:
\begin{subequations}
    \begin{align}
        \Sigma_t^{i-}&=F^i\Sigma_t^-F^{i\top},\\
        \Sigma_t^i&=(I-K_t^iC_t^i)\Sigma_t^{i-}(I-K_t^iC_t^i)^\top+K_t^iV_t^iK_t^{i\top},\label{eqn:joseph}
    \end{align}
\end{subequations}
where $F^i\coloneqq[\textbf{0},\cdots,\textbf{0},I_n,\textbf{0},\cdots,\textbf{0}]$
with its $i^\text{th}$ block being $I_n$. Note that \eqref{eqn:joseph} naturally holds since $\mathcal{K}_t$ and $\mathbb{K}_t$ are block diagonal matrices.
By letting $\frac{\partial \textbf{tr}(\Sigma_t^i)}{\partial K_t^i}=\textbf{0}$, the stepwise optimal individual estimation gain is expressed as follows:
\begin{equation}
\begin{aligned}
    &K_t^{i\star}(C_t^i\Sigma_t^{i-}C_t^{i\top}+V_t^i)=\Sigma_t^{i-}C_t^i,\\
    &K_t^{i\star}=\Sigma_t^{i-}C_t^i(C_t^i\Sigma_t^{i-}C_t^{i\top}+V_t^i)^{-1}.
\end{aligned}
    \label{eqn:alternating_kalman}
\end{equation}
Note that each agent's estimation gain can be optimized independently. However, this gain generally processes all sensor observations, motivating a sparse relaxation.

\subsection{Sparse Estimation via Group Lasso Optimization}
\label{sec:sparse_est_group_lasso}

We therefore seek a sparse approximation to $K_{t+1}^{i\star}$ that zeroes out columns corresponding to uninformative sensors, allowing each agent to rely on fewer observations while preserving estimation accuracy. Mirroring the partition of $\textbf{y}_t^i$ into $P_i$ sensor groups, we partition $K_t^i$ into $P_i$ block columns as $K_t^i=[K_t^i[1],\cdots,K_t^i[P_i]]$, where $K_t^i[\rho]$ is the estimation gain associated with sensor $\rho$. Subsequently, we formulate the following group lasso problem for all $i\in[N],\ t\in[T+1]$:
\begin{equation}
        \hspace{-1mm}\min_{K_t^i}\ \underbrace{\left\|K_t^i(C_t^i\Sigma_t^{i-}C_t^{i\top}+V_t^i)-\Sigma_t^{i-}C_t^i\right\|_F^2}_{L_{\text{Estimation}}}+\underbrace{\sum_{\rho\in[P_i]}\lambda_t^{i\rho}\left\|K_t^i[\rho]\right\|_F}_{L_{\text{Sparsity}}},
    \label{eqn:lasso_sparse_estimation}
\end{equation}
where $\|\cdot\|_F$ is the matrix Frobenius norm. In this optimization problem, $L_{\text{Estimation}}$ denotes a least square error cost that encourages the stepwise estimation gain $K_t^i$ to be closer to the optimal individual estimation gain in \eqref{eqn:alternating_kalman}. On the other hand, $L_{\text{Sparsity}}$ penalizes large norms of $K_t^i[\rho]$, with non-negative regularization constant $\lambda_t^{i\rho}$, encouraging agent $i$ not to rely on observation $\textbf{y}_t^{i\rho}$, i.e., the measurement from the $\rho$-th sensor. 

In what follows, we show that \eqref{eqn:lasso_sparse_estimation} can be cast as a convex conic program, solvable by off-the-shelf optimizers \cite{Clarabel_2024}.
\begin{proposition}
Problem \eqref{eqn:lasso_sparse_estimation} is equivalent to a convex conic program with quadratic cost.
\end{proposition}
\begin{proof}
Using matrix vectorization \cite{macedo2013typing}, the group lasso problem \eqref{eqn:lasso_sparse_estimation} can be reconstructed as
\begin{multline}
        \min_{K_t^i}\left\|\left((C_t^i\Sigma_t^{i-}C_t^{i\top}+V_t^i)\otimes I\right)\text{vec}(K_t^i)-\text{vec}(\Sigma_t^{i-}C_t^i)\right\|_2^2\\
        +\sum_{\rho\in[P_i]}\lambda_t^{i\rho}\|\text{vec}(K_t^i[\rho])\|_2,\label{eqn:vec_group_lasso}
\end{multline}
where $\otimes$ denotes the Kronecker product. Moreover, introducing the slack variables $s_t^{i\rho}$, it follows that the optimization \eqref{eqn:vec_group_lasso} has the same optimal solution as 
\begin{align*}
        \min_{K_t^i, s_t^i}&\left\|\left((C_t^i\Sigma_t^{i-}C_t^{i\top}{+}V_t^i)\otimes I\right)\text{vec}(K_t^i){-}\text{vec}(\Sigma_t^{i-}C_t^i)\right\|_2^2{+}\sum_{\rho\in[P_i]}\lambda_t^{i\rho}s_t^{i\rho},\label{eqn:second-order_cone_optimization}\\
        \text{s.t.}\ &\begin{bmatrix}
            \text{vec}(K_t^i[\rho]) \\
            s_t^{i\rho}
        \end{bmatrix}\in\mathfrak{C},\ \mathfrak{C}\coloneqq\left\{(\mathfrak{p},\mathfrak{q})\left|\ \mathfrak{q}\in\mathbb{R}, \|\mathfrak{p}\|_2\leq \mathfrak{q}\right\}\right.,
\end{align*}
which is a convex conic program with a second-order cone constraint by definition \cite{boyd2004convex}. 
\end{proof}
Finally, note that \eqref{eqn:lasso_sparse_estimation} continuously optimizes the linear estimation gain matrices $K_t^i$ for all players and time steps. In practice, however, each agent $i$'s ultimate goal is to have a binary decision about whether to utilize sensor $\rho$ to obtain the corresponding measurement $\textbf{y}_t^{i\rho}$. If $\textbf{y}_t^{i\rho}$ is not used in constructing $\hat{\textbf{x}}_t^i$, it is equivalent to \textbf{zeroing out} $K_t^i[\rho]$, the $\rho$-th block column of the linear estimation gain $K_t^i$.
In that way, we introduce a ratio threshold $r_{\text{th}}\in(0,1]$ to zero out $K_t^i[\rho]$ if its norm $\|K_t^i[\rho]\|_F$ is  small enough compared to the optimal dense estimation gain $K_t^{i\star}[\rho]$ in \eqref{eqn:alternating_kalman}, i.e.,
\begin{equation}
K_t^i[\rho] =
\begin{cases}
K_t^{i\star}[\rho], & \|K_t^i[\rho]\|_F\geq r_{\text{th}}\|K_t^{i\star}[\rho]\|_F,\quad [\textbf{reset}]\\
\mathbf{0}, & \text{otherwise}.\qquad\qquad\quad\  [\textbf{zero-out}]
\end{cases}
\label{eqn:reset_rule}
\end{equation}
The resulting procedure is summarized in Algorithm~\ref{alg:LQG_game_play}. 
\input{algorithms/lqg_game_play}
  
\subsection{Game-Theoretic Control-Adaptive Regularization}

For the remainder of this section, we specialize to the interagent measurement case, where each sensor directly measures a specific agent's state ($C_t^i = I$). That is, each measurement $\textbf{y}_t^{ij}$ represents agent $i$'s noisy observation of agent $j$. In this setting, the sparse estimation framework directly determines \textit{which agents} each agent observes, making the connection between estimation sparsity and interagent communication explicit.

Since the sparsity decision now corresponds to which agents to observe, the regularization $\lambda_t^{ij}$ should reflect how much agent $i$'s control depends on agent $j$'s state. This dependency is captured by $\Gamma_t^i[j]$, the $j$-th block of the feedback gain. We thus propose the following control-adaptive regularization:
\begin{equation}
\lambda_{t+1}^{ij} =
\begin{cases}
\frac{L_1}{\|\Gamma_t^i[j]\|_F}, & \|\Gamma_t^i[j]\|_F\neq 0,\\
L_2, & \text{otherwise},
\end{cases}
\label{eqn:adaptive_lambda}
\end{equation}
where $\Gamma_t^i[j]$ is the $j$-th block of $\Gamma_t^i$, and $L_1, L_2 > 0$ are tuning parameters. When agent $i$'s control depends strongly on agent $j$ (large $\|\Gamma_t^i[j]\|_F$), the regularization is low and the observation is (more likely) retained; when the dependency is weak or absent, $\lambda_t^{ij}$ defaults to $L_2$, promoting sparsification.

\subsection{Covariance-Sparsity Trade-Off}

A natural question is how much estimation accuracy is sacrificed by sparsifying the estimation gain. We address this by providing a sufficient condition on the regularization level under which the reset mechanism is guaranteed to trigger. This condition becomes easier to satisfy as the estimation error covariance grows, implying that the system self-corrects before estimation quality degrades too far.

\begin{theorem}
    In the interagent measurement case, at time $t$, $K_t^i[\rho]$ will reset to $K_t^{i\star}[\rho]$ for $\rho\in[P_i]$ if 
    \begin{equation*}
        \max_{\rho\in[P_i]}\lambda_t^{i\rho}\leq \frac{2(1-r_{\text{th}})}{\sqrt{P_i}}\cdot\frac{\sigma_\text{min}(\Sigma_t^{i-})[\sigma_\text{min}(\Sigma_t^{i-})+\sigma_\text{min}(V_t^i)]^2}{\sigma_\text{max}(\Sigma_t^{i-})+\sigma_\text{max}(V_t^i)},
    \end{equation*}
    where $\sigma_\text{min}(\cdot)$ and $\sigma_\text{max}(\cdot)$ denote the minimum and maximum singular value for a matrix, respectively.
\end{theorem}

\begin{proof}
    \textbf{Case 1:} Suppose $K_t^i[\rho]\neq \textbf{0}, \forall\rho\in[P_i]$. Then, the first-order optimality condition for the problem \eqref{eqn:lasso_sparse_estimation} gives:
    \begin{equation*}
        \begin{aligned}
            &2(K_t^i-K_t^{i\star})(\Sigma_t^{i-}+V_t^i)^2=-\begin{bmatrix}
                \frac{\lambda_t^{i1}K_t^i[1]}{\|K_t^i[1]\|_F} & \cdots & \frac{\lambda_t^{iP_i}K_t^i[P_i]}{\|K_t^i[P_i]\|_F}
            \end{bmatrix}\\
            &\Rightarrow K_t^i=K_t^{i\star}+K_\lambda,
        \end{aligned}
    \end{equation*} 
    where $K_\lambda=-\frac{1}{2}\begin{bmatrix}
                \frac{\lambda_t^{i1}K_t^i[1]}{\|K_t^i[1]\|_F} & \cdots & \frac{\lambda_t^{iP_i}K_t^i[P_i]}{\|K_t^i[P_i]\|_F}
            \end{bmatrix}(\Sigma_t^{i-}+V_t^i)^{-2}$.
   Let $\Bar{\lambda}_t^i=\max_{\rho\in[P_i]}\lambda_t^{i\rho}$. Denoting the $\rho$-th block column of $K_\lambda$ as $K_\lambda[\rho]$, it follows that
    \begin{equation}
        \begin{aligned}
            \|K_\lambda[\rho]\|_F\leq\|K_\lambda\|_F&\leq 
            \frac{\sqrt{P_i}\Bar{\lambda}_t^i}{2[\sigma_\text{min}(\Sigma_t^{i-}+V_t^i)]^2}.\label{eqn:k_lambda_bound}
        \end{aligned}
    \end{equation}
    Since $K_t^i = K_t^{i\star} + K_\lambda$ holds block-column-wise, i.e., $K_t^i[\rho] = K_t^{i\star}[\rho] + K_\lambda[\rho]$, the reverse triangle inequality and $\sigma_\text{min}(\Sigma_t^{i-}+V_t^i)\geq\sigma_\text{min}(\Sigma_t^{i-})+\sigma_\text{min}(V_t^i)$ give
    \begin{align}\nonumber
            \|K_t^i[\rho]\|_F&\geq \|K_t^{i\star}[\rho]\|_F-\|K_\lambda[\rho]\|_F\\
            &\geq \|K_t^{i\star}[\rho]\|_F-\frac{\sqrt{P_i}\Bar{\lambda}_t^i}{2[\sigma_\text{min}(\Sigma_t^{i-})+\sigma_\text{min}(V_t^i)]^2}.\label{eq:ident}
     \end{align}
    Next, note that for $K_t^i[\rho]$ to reset to $K_t^{i\star}[\rho]$, the rule \eqref{eqn:reset_rule}  requires $\|K_t^i[\rho]\|\geq r_{\text{th}}\|K_t^{i\star}[\rho]\|_F$. Combining this rule with \eqref{eq:ident}, a sufficient condition for $K_t^i[\rho]$ to reset to $K_t^{i\star}[\rho]$ is given by
    \begin{multline}\label{eq:fin1}
        \|K_t^{i\star}[\rho]\|_F-\frac{\sqrt{P_i}\Bar{\lambda}_t^i}{2[\sigma_\text{min}(\Sigma_t^{i-})+\sigma_\text{min}(V_t^i)]^2}\geq r_{\text{th}}\|K_t^{i\star}[\rho]\|_F \\       
        \Leftrightarrow\ \  (1-r_{\text{th}})\|K_t^{i\star}[\rho]\|_F\ge\frac{\sqrt{P_i}\Bar{\lambda}_t^i}{2[\sigma_\text{min}(\Sigma_t^{i-})+\sigma_\text{min}(V_t^i)]^2}.
    \end{multline}
    Denote $K_t^{i\star}[\rho]=K_t^{i\star}E_\rho$, where $E_\rho$ is a matrix that selects the $\rho$-th block column of $K_t^{i\star}$. Then, it follows that
    \begin{multline}\label{eq:fin2}
        \|K_t^{i\star}[\rho]\|_F\geq\|K_t^{i\star}[\rho]\|_2\geq\sigma_\text{min}(K_t^{i\star})\|E_\rho\|_2\\= \sigma_\text{min}(K_t^{i\star})\geq\frac{\sigma_\text{min}(\Sigma_t^{i-})}{\sigma_\text{max}(\Sigma_t^{i-})+\sigma_\text{max}(V_t^i)}.
    \end{multline}
    Combining \eqref{eq:fin1}-\eqref{eq:fin2}, it follows that a sufficient condition for $K_t^{i}[\rho]$ to reset to $K_t^{i\star}[\rho]$ is
    \begin{equation}
        \Bar{\lambda}_t^i\leq \frac{2(1-r_{\text{th}})}{\sqrt{P_i}}\cdot\frac{\sigma_\text{min}(\Sigma_t^{i-})[\sigma_\text{min}(\Sigma_t^{i-})+\sigma_\text{min}(V_t^i)]^2}{\sigma_\text{max}(\Sigma_t^{i-})+\sigma_\text{max}(V_t^i)}\label{eqn:sufficient}
    \end{equation}
    \textbf{Case 2:} Suppose $K_t^i[\rho]=\textbf{0}$ is the optimal solution for \eqref{eqn:lasso_sparse_estimation} for some $\rho$. Then, using the subgradient optimality condition in \eqref{eqn:lasso_sparse_estimation} with respect to $K_t^i[\rho]$, we obtain
\begin{equation}
    \begin{aligned}
        \lambda_t^{i\rho}&\geq\left\|2K_t^{i\star}(\Sigma_t^{i-}+V_t^i)^2E_\rho\right\|_F\geq 2\sigma_\text{min}\!\left(\Sigma_t^{i-}(\Sigma_t^{i-}+V_t^i)\right)\\
        &\geq 2\sigma_\text{min}(\Sigma_t^{i-})\left[\sigma_\text{min}(\Sigma_t^{i-})+\sigma_\text{min}(V_t^i)\right]\\
        &> \frac{2(1-r_{\text{th}})}{\sqrt{P_i}}\cdot\frac{\sigma_\text{min}(\Sigma_t^{i-})\left[\sigma_\text{min}(\Sigma_t^{i-})+\sigma_\text{min}(V_t^i)\right]^2}{\sigma_\text{max}(\Sigma_t^{i-})+\sigma_\text{max}(V_t^i)},
    \end{aligned}
    \label{eqn:k=0}
\end{equation}
    where the last inequality follows from $\sqrt{P_i}\geq 1$, $1-r_{\text{th}}<1$, and $\frac{\sigma_\text{min}(\Sigma_t^{i-})+\sigma_\text{min}(V_t^i)}{\sigma_\text{max}(\Sigma_t^{i-})+\sigma_\text{max}(V_t^i)}\leq 1$. This contradicts \eqref{eqn:sufficient}. Hence, $K_t^i[\rho]=\textbf{0}$ cannot be optimal under \eqref{eqn:sufficient}, hence \eqref{eqn:sufficient} is sufficient to trigger a reset of $K_t^i[\rho]$ to $K_t^{i\star}[\rho]$.
\end{proof}

%% file: algorithms/lqg_game_play.tex
\begin{algorithm}[!h]
    \small
    \caption{LQG Game with Sparse Distributed Estimation}
    \begin{algorithmic}[1]
        \renewcommand{\algorithmicrequire}{\textbf{Input:}}
        \renewcommand{\algorithmicensure}{\textbf{Output:}}
        \REQUIRE 
        \textbf{System Parameters:} ${\Sigma_{\textbf{x},1}^-, (A_t,B_t,W_t)_{t=1}^T, (\{C_t^i\}_{i=1}^N,V_t)_{t=1}^{T+1}}$,\\
        \ \ \quad \textbf{Deterministic Feedback Nash Strategy:} $(\Gamma_t,\alpha_t)_{t=1}^{T}$,\\
        \ \  \quad \textbf{Regularization level} $(\{\lambda_t^{i\rho}\}_{\rho=1}^{P_i})_{i=1,t=1}^{N,T+1}$.\\
        \ENSURE Sparse Individual Estimation Gain $\{K_t^i\}_{i=1,t=1}^{N,T+1}$.
        \FOR{$t=1:T$}
        \FOR{$i=1:N$}
            \STATE $K_t^i\leftarrow$ \eqref{eqn:lasso_sparse_estimation}\quad \texttt{// Group lasso optimization} \\
            \STATE $\textbf{y}_t^i\leftarrow$ \eqref{eqn:ind_obs}\quad \texttt{// Individual noisy observation}
            \STATE $\hat{\textbf{x}}_t^i\leftarrow$ \eqref{eqn:est_x_update}\quad \texttt{// Distributed \textbf{sparse} estimation}
            \STATE $\hat{\textbf{u}}_t^i\leftarrow$ \eqref{eqn:estimated_control}\quad \texttt{// Feedback control estimation}
        \ENDFOR
        \STATE $\textbf{u}_t\leftarrow$ \eqref{eqn:actual_control}\quad \texttt{// Actual feedback joint control}
        \STATE $\textbf{x}_{t+1}\leftarrow$ \eqref{eqn:x_update}~ \texttt{// System update with process noise}
        \ENDFOR
    \end{algorithmic}
    \label{alg:LQG_game_play}
\end{algorithm}

%% file: experiment.tex
\label{sec:experiment}
We validate our proposed sparse LQG games framework in a three-robot formation game, and investigate robots' trajectory performance under different regularization designs.

\begin{figure*}[!t]
    \centering
    \includegraphics[width=0.82\linewidth]{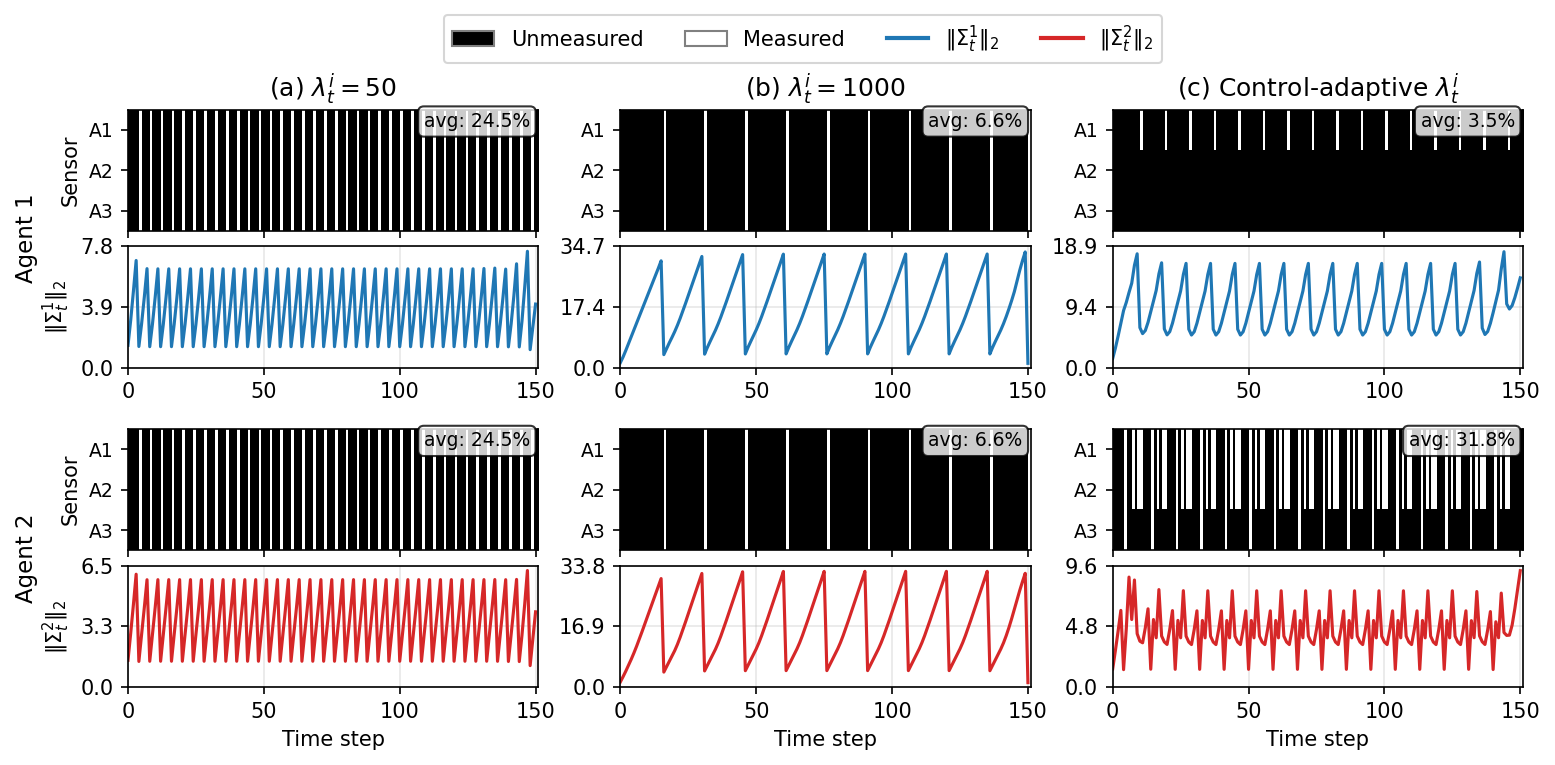}
    \caption{The sensor usage and individual estimation covariance norm evolution of each robot (Row 1: R1, Row 2: R2, we omit R3 as it is symmetric to R2) in the three-robot formation game under constant regularization levels (a) $\lambda_t^i=50$, (b) $\lambda_t^i=1000$ and (c) game-theoretic control-adaptive $\lambda_t^i$.}
    \label{fig:sensor_and_norm}
    \vspace{-4ex}
\end{figure*}

\begin{figure}[!t]
    \centering
    \includegraphics[width=0.99\linewidth]{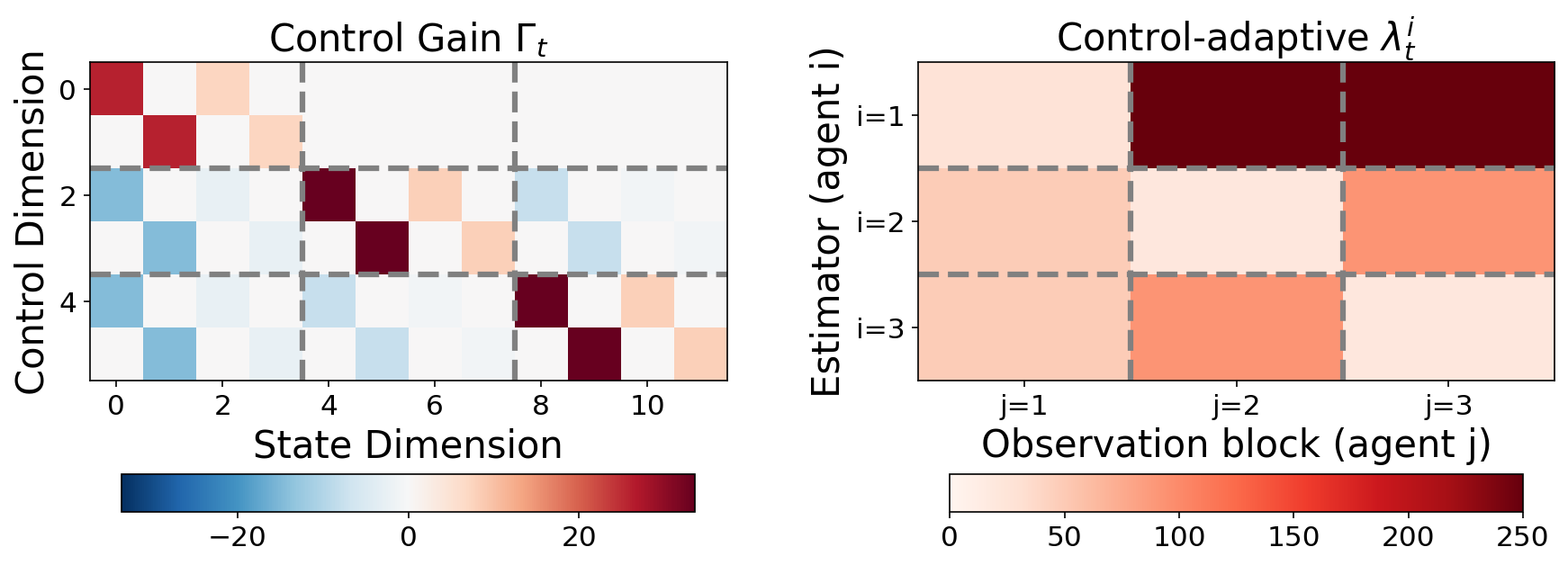}
    \caption{Game-theoretic control-adaptive regularization with respect to feedback Nash control gain $\Gamma_t^i[j]$ in \eqref{eqn:adaptive_lambda} for a three-robot formation game.}
    \label{fig:adaptive_lambda}
    \vspace{-4ex}
\end{figure}

\subsection{Simulation Setup}
The formation game comprises a leading robot (R1) that aims to track a static reference path $p_{t,\text{ref}}^1$, and following robots (R2 and R3) that aim to maintain a reference displacement $d_{\text{ref}}^{ij}$ between them. Each robot follows double-integrator dynamics discretized with time-step $\Delta t$, i.e.,
\begin{equation*}
    A_t^i=\begin{bmatrix}
        I_2 & \Delta t\cdot I_2 \\ \textbf{0} & I_2
    \end{bmatrix},\  
    B_t^i=\begin{bmatrix}
        \frac{\Delta_t^2}{2}\cdot I_2 \\ \Delta t\cdot I_2
    \end{bmatrix},
\end{equation*}
with state $x_t^i=[p_{t,x}^i,p_{t,y}^i,v_{t,x}^i,v_{t,y}^i]^\top$ encoding 2-D position and velocity and control $u_t^i=[a_{t,x}^i,a_{t,y}^i]^\top$ encoding 2-D acceleration. We formulate each robot's objective as follows:
\begin{equation}
    \begin{aligned}
        J^1&=\sum_{t=1}^T\mathbb{E}\left[\|p_t^1-p_{t,\text{ref}}^1\|_2^2+\|u_t^1\|_2^2\right]\notag+\mathbb{E}\left[\|p_{T+1}^1-p_{T+1,\text{ref}}^1\|_2^2\right],\\
        J^i&=\sum_{t=1}^T\sum_{j\neq i}\mathbb{E}\left[\|p_t^i-p_t^j-d_{\text{ref}}^{ij}\|_2^2+\|u_t^i\|_2^2\right]\notag\\
        &+\sum_{j\neq i}\mathbb{E}\left[\|p_{T+1}^i-p_{T+1}^j-d_{\text{ref}}^{ij}\|_2^2\right], i\in\{2,3\}, j\in\{1,2,3\},
    \end{aligned}
\end{equation}
with $T=150, \Delta t=\SI{0.05}{s}$, the reference displacement between robots $d_{\text{ref}}^{21}=[-2;-2], d_{\text{ref}}^{31}=[2;-2], d_{\text{ref}}^{23}=-d_{\text{ref}}^{32}=[4;0]$, and the reference path for R1: $p_{t,\text{ref}}^1=\big(30\cos{(k\cdot\Delta t)}, 30\sin{(3k\cdot\Delta t)}\big)$ for $k\in[T]$. We also set  $
        \Sigma_{\textbf{x},1}^-=I_3\otimes\text{blockdiag}(0.25^2I_2, 1.2^2I_2),$ $
        W_t^i=\text{blockdiag}(\textbf{0}_{2\times 2}, 1.2^2I_2),$ $
        V_t^i=I_3\otimes\text{blockdiag}(0.25^2I_2, 1.2^2I_2)$
and $C_t^i=I_{12}, \forall i\in[3]$ being identical across robots. For the control-adaptive regularization design in \eqref{eqn:adaptive_lambda}, we set $L_1=1000$ and $L_2=250$, and the resulting regularization is illustrated in Fig. \ref{fig:adaptive_lambda}.

\begin{figure*}[!t]
    \centering
    \includegraphics[width=0.74\linewidth]{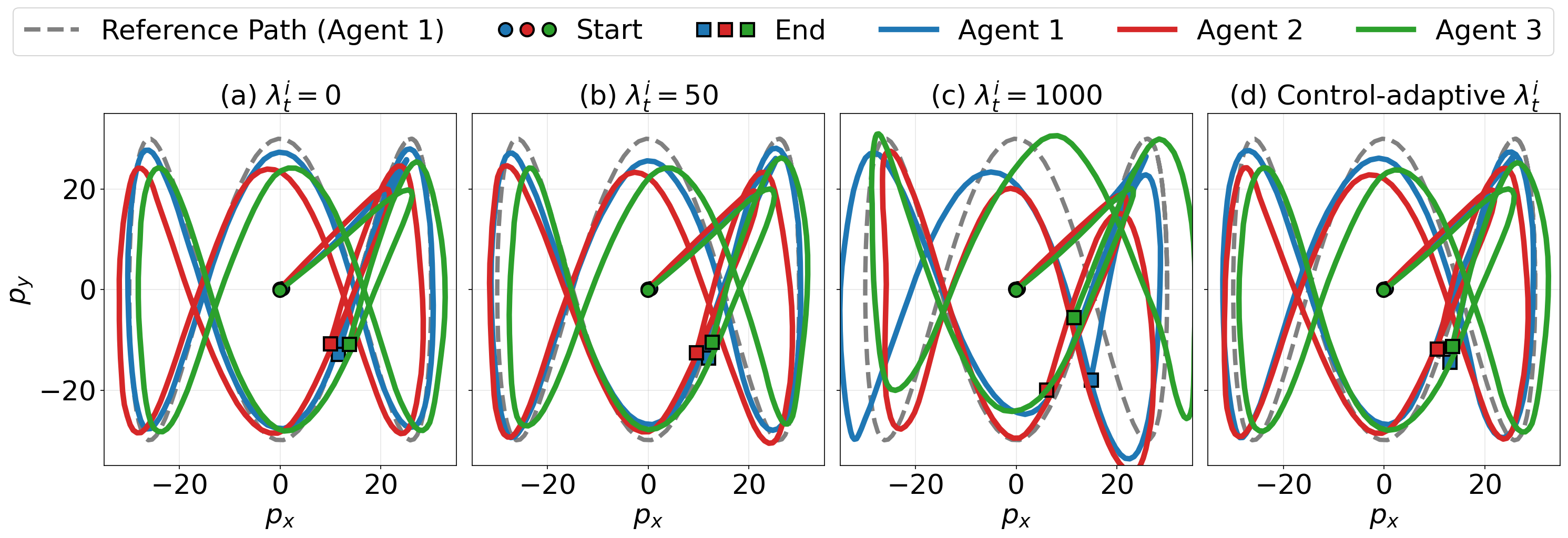}
    \caption{Trajectory performance for the formation game with regularization levels (a) $\lambda_t^i=0$, (b) $\lambda_t^i=50$, (c) $\lambda_t^i=1000$, and (d) control-adaptive $\lambda_t^i$.}
    \label{fig:trajectory_comparison}
    \vspace{-4ex}
\end{figure*}

\subsection{Result Analysis}
We evaluate our distributed sparse estimation with game-theoretic control-adaptive regularization and compare it with static regularization levels \(\lambda_t^i=50\) and \(1000\). We observe the following behaviors: \romannumeral1) Fig.~\ref{fig:sensor_and_norm} shows that sensor usage decreases as the regularization level increases, since \(L_{\text{Sparsity}}\) becomes more dominant in \eqref{eqn:lasso_sparse_estimation}. \romannumeral2) Fig.~\ref{fig:trajectory_comparison} further shows that stronger regularization reduces sensing usage but degrades control performance because estimation error accumulates more quickly, highlighting the trade-off between control performance and estimation sparsity. \romannumeral3) The estimation error increases when measurements are skipped and drops sharply when dense measurements are taken, suggesting that the covariance remains bounded under sparse estimation.

Compared with static regularization, the proposed control-adaptive regularization achieves a better balance between sensing resource allocation and control performance. As shown in Fig.~\ref{fig:sensor_and_norm}, R1 saves sensing resources by prioritizing self-measurement, while R2 and R3 still frequently measure all robots. The sensor usage pattern also shows that R2 and R3 measure R1 more often than they measure each other, reflecting R1's leadership role in the game. Fig.~\ref{fig:trajectory_comparison} shows that this adaptive strategy allows R1 to track the reference path well with limited sensing, while R2 and R3 follow accordingly. Overall, compared with static regularization, control-adaptive regularization leads to a more effective allocation of inter-robot measurements in games with structured coupling among agents' objectives.

%% file: conclusion.tex
\label{sec:conclusion}

We studied LQG games in which multiple agents estimate the joint system state using sparse individual observations. To promote sparsity in each agent’s estimator, we proposed a group lasso formulation that zeros out groups of estimation gains, yielding a convex optimization problem.  Future work includes an extension of the proposed approach to nonlinear dynamics and non-quadratic agent objectives, via iterative LQ game approaches such as \cite{fridovich2020efficient,wang2020game,schwarting2021stochastic}.